\documentclass[preprint,12pt,english]{elsarticle}
\usepackage{amsmath}
\usepackage{amssymb}
\usepackage{lipsum}
\usepackage[sans]{dsfont}
\usepackage{multirow}
\usepackage{afterpage}
\usepackage{dsfont}
\usepackage{amsfonts}
\usepackage{bm}%
\usepackage{babel}
\usepackage{appendix}
\usepackage{hyperref}
\usepackage{graphicx}
\usepackage{color}
\usepackage{amsthm}
\def\fnum@figure{\figurename\thefigure}
\renewcommand{\figurename}{Fig.}

\newcommand{\sech}{\textrm{sech}}
\newtheorem{theorem}{Theorem}

\newtheorem{lemma}[theorem]{Lemma}
\journal{Physica A}

\title{Measurement-induced nonlocality quantified by Hellinger distance and weak measurements}
\begin{document}
\author{Indrajith V S$^*$ , R. Muthuganesan$^\ddag$ R. Sankaranarayanan$^*$ }
\address{$^*$ Department of Physics, National Institute of Technology\\ Tiruchirappalli - 620015, Tamil Nadu, India.}
\address{$^\ddag$ Centre for Nonlinear Science and Engineering, School of Electrical and Electronics Engineering,  SASTRA Deemed University, Thanjavur - 613401, Tamil Nadu, India.}
\begin{abstract}
In this article, we propose measurement-induced nonlocality (MIN) quantified by Hellinger distance using von Neumann projective measurement. The proposed MIN is a bonafide measure of nonlocal correlation and is resistant to local ancilla problem. We obtain an analytical expression of Hellinger distance MIN for general pure and $2 \times n$ mixed states. In addition to comparing with similar measures, we explore the role of weak measurement in capturing nonlocal correlation.

\end{abstract}


\maketitle

\section{Introduction}
Strange correlation between different parts of a quantum system that cannot be described by a local hidden variable theory, and violates Bell inequality is referred to as nonlocality \cite{bel}. Nonlocal correlation is one of the key features indeparting quantum system from the corresponding classical counterpart. Quantification of nonlocal correlation takes quantum information processing to next level and is used in various areas such as teleportation \cite{telepor}, communication etc. Quantification of nonlocal correlation had been centred around entanglement for quite a long time until the introduction of quantum discord \cite{discord}. This entropic measure captures nonlocality beyond entanglement and Bell’s test of inequality. Later development in this direction happened when the geometric measure of quantum discord \cite{geo_discord} was proposed, as it is much easier to compute than the entropic measure. Lu and Fu introduced measurement-induced nonlocality (MIN) \cite{min} which is also a geometric measure of nonlocal correlation.

MIN characterizes the nonlocal behaviour by employing statistical distance between a quantum state and its local invariant projective measured state. This quantity can be regarded as complementary to the geometric measure of quantum discord. Though MIN could not resolve the local ancilla problem \cite{piani}, it can be resolved by replacing density matrix with its square root \cite{tracesqroot}. Different forms of MIN have also been investigated using skew information \cite{minskew}, trace distance \cite{tmin}, von Neumann entropy \cite{von_min}, relative entropy \cite{rel_entropy_min}, fidelity \cite{fmin} etc. A handful of studies of MINs are made on Heisenberg spin models and noisy channels to show that MINs demonstrate nonlocal correlation in bipartite states even without entanglement \cite{muthu1, my}.

In this article, we propose a MIN quantified by Hellinger distance (H-MIN). Similar to the other MINs, this is also a distance measure taken between pre- and post-measurement states based on von Neumann projective measurement. Being a bonafide measure of nonlocal correlation, it resolves the local ancilla problem. In addition to demonstrating H-MIN for few well known bipartite states, its direct connection to other MINs are also shown.

The notion of weak measurements was introduced by Aharonov, Albert and Vaidman (AAV) \cite{weak1} in 1989, wherein the measurements are done without complete disintegrating of a quantum state to its eigenstate. These kinds of measurements possess some strange behaviour which AAV explains in their article. The notion of weak measurement is considered in the context of nonlocal correlation by proposing super quantum discord \cite{super_discord}, which possesses a higher level of correlation than quantum discord. Later, weak measurements on geometric discord \cite{weak_geo, coh_review} and quantum correlations \cite{weak_corltn, weak_corltn2} are introduced, from which geometric discord is obtained as a special case. In light of this, here we introduce the notion of sequential weak measurement and its connection to the post-measured state.
\section{Measurement-Induced Nonlocality (MIN)}
It is a correlation measure of bipartite quantum state in the geometric perspective to capture non-local effect due to invariant local projective measurements \cite{min}. This quantity may be considered as dual to geometric quantum discord \cite{geo_discord}, and is defined as 
\begin{equation}
N_2(\rho):=~ ^{\text{max}}_{\Pi^a}\Vert\rho-\Pi^a(\rho)\Vert^2
\end{equation}
where $\Vert \mathcal{O} \Vert = \sqrt{\text{tr}(\mathcal{O}^{\dagger}\mathcal{O})}$ is the Hilbert-Schmidt norm of an operator $\mathcal{O}$. Here the maximum is taken over all possible locally invariant von Neumann projective measurements $\Pi^a=\lbrace\Pi^a_k\rbrace = \lbrace\vert k \rangle \langle k \vert \rbrace ~\text{and}\,\Pi^a(\rho) = \sum_k(\Pi_k^a \otimes \mathbb{I}^b)\rho(\Pi_k^a \otimes \mathbb{I}^b)$.

An arbitrary state of a bipartite $m \times n$ dimensional composite system can be written as
\begin{equation}
\rho= \sum_{i,j}\gamma' _{ij}X_{i}\otimes  Y_{j} \label{state1}.
\end{equation}
Here \{$X_i : i = 1, 2, \cdots, m^2\}$ and \{$Y_j : j = 1, 2, \cdots, n^2\}$ are set of self-adjoint orthonormal operators on Hilbert spaces $\mathcal{H}^a ~\text{and} ~\mathcal{H}^b$ respectively such that $\text{tr}(X_i^{\dagger}X_j) = \text{tr}(Y_i^{\dagger}Y_j)= \delta_{ij}$ with $X_1 = \mathbb{I}/\sqrt{m}$, $Y_1 = \mathbb{I}/\sqrt{n}$, and $\gamma' _{ij} =\hbox{tr} (\rho\,X_{i}\otimes Y_{j})$ is an element of $m^2 \times n^2$ real matrix. The state $\rho$ can also be written as  
\begin{equation}
\rho =\frac{1}{\sqrt{m n}}\frac{\mathbb{I}^{a}}{\sqrt{m}}\otimes \frac{\mathbb{I}^{b}}{\sqrt{n}}+\sum_{i=2}^{m^2 }x_{i}X_{i}\otimes\frac{\mathbb{I}^{b}}{\sqrt{n}}+\frac{\mathbb{I}^{a}}{\sqrt{m}}\otimes\sum_{j=2}^{n^2 }y_{j}Y_{j} +\sum_{i,j\neq 1}t_{ij }X_{i}\otimes Y_{j} \label{state2}
\end{equation}
where   $x_{i}=\hbox{tr}(\rho\,X_{i}\otimes \mathbb{I}^{b})/\sqrt{m}$, $y_{j}= \hbox{tr}(\rho\,\mathbb{I}^{a} \otimes Y_{j} )/\sqrt{n}$ and $T = (t _{ij} = \hbox{tr} (\rho\,X_{i}\otimes  Y_{j}))$ is a real correlation matrix of order $(m^{2}-1)\times(n^2 -1)$.

In fact, MIN is easy to compute and it has a closed formula for  $2\times n$ dimensional system as
\begin{equation}
N_{2}(\rho)= \begin{cases}
\text{tr}(TT^t)-\frac{1}{\Vert \textbf{x} \Vert^2}\textbf{x}^tTT^t\textbf{x} &  \text{if} \quad \textbf{x}\neq 0, \\
\text{tr}(TT^t)-\lambda_{\text{min}} &  \text{if} \quad \textbf{x}=0, 
\end{cases}
\end{equation}
where $\lambda_{\text{min}}$ is the least eigenvalue of $3\times 3$ dimensional matrix $TT^t$ and $\textbf{x} = (x_1~x_2~ x_3)^t $. 
\section{MIN Based on Hellinger Distance}
Hellinger distance is another useful geometric measure between two probability distributions. Replacing the probability distributions by density matrices, it becomes a metric in state space. It is a useful quantifier for the manifestation of quantum nonlocality. Hellinger distance \cite{heli_dist} between two states $\rho$ and $\sigma$ is given as $\mathcal{D}_H = \text{tr} (\sqrt{\rho}-\sqrt{\sigma})^2$.

Considering a bipartite system (\ref{state1}), the measurement-induced nonlocality based on Hellinger distance (H-MIN) is defined as
\begin{equation}
N_{H}(\rho) :=~^\text{max}_{\Pi^{a}}\lVert \sqrt{\rho}-\Pi^{a}(\sqrt{\rho})\rVert^2\label{first}
\end{equation}
or equivalently
\begin{equation}
N_{H}(\rho) :=~^\text{max}_{\Pi^{a}}~\text{tr}~\left(\sqrt{\rho}-\Pi^{a}(\sqrt{\rho})\right)^2
\end{equation}
where the maximum is taken over locally invariant von Neumann projective measurements on subsystem $a$. With this definition, the H-MIN can then be computed as
\begin{equation}
N_{H}(\rho) = \left(1- ~^\text{min}_{\Pi^{a}}~\text{tr}[\sqrt{\rho}~\Pi^{a}(\sqrt{\rho})]\right)\label{hmin}
\end{equation}
whose properties are as given below:
\begin{enumerate}
  \item $N_H(\rho) \geq 0$, with equality holds only for product state $\rho = \rho_a \otimes \rho_b$.
  \item $N_H$ is invariant under the addition of local ancilla $\rho^c$. Defining $ \rho^{a:bc} = \rho^{ab}\otimes \rho^c$, the H-MIN is
  \begin{align*}
    N_H(\rho^{a:bc}) &= \,^{\text{max}}_{\Pi^a}\,\text{tr}\bigg(\sqrt{\rho^{ab}\otimes\rho^c}-\Pi^a\left(\sqrt{\rho^{ab}\otimes\rho^c}\right)\bigg)^2\\
    &=\, ^{\text{max}}_{\Pi^a}\,\text{tr}\bigg(\sqrt{\rho^{ab}}-\Pi^a\left(\sqrt{\rho^{ab}}\right)\bigg)^2 \cdot \text{tr}(\rho^c)
    \end{align*}
     where we use the property of trace operation $\text{tr}(a \otimes b) = \text{tr}(a) \cdot \text{tr}(b)$ in the last step. It implies that
    \begin{align*}
    N_H(\rho^{a:bc}) &= N_H(\rho^{ab})\cdot\text{tr}(\rho^c) \\
    &= N_H(\rho^{ab}) 
    \end{align*}   
and thus fixing the local ancilla problem\cite{piani}.
  \item For any unitary operators $U$ and $V$, $N_H \big((U \otimes V)\rho(U \otimes V)^\dagger\big) =     N_H(\rho)$ such that $N_H (\rho)$ is invariant under local unitary operation. To prove this, it is enough to show that $\text{tr}[\sqrt{\rho}\,\Pi^a(\sqrt{\rho})]$ is invariant under local unitary operations as shown below. Using the cyclic property of trace, the following quantity 
 \end{enumerate}
   \begin{align*}
    \text{tr} [(U\otimes V)\sqrt{\rho}(U\otimes V)^\dagger\,\Pi^a \otimes \mathbb{I}\,(U\otimes V)\sqrt{\rho}(U\otimes V)^\dagger\,\Pi^a \otimes \mathbb{I}]
    \end{align*}
    is equal to
     \begin{align*}
     \text{tr} [(U\otimes V)^\dagger\,\Pi^a \otimes \mathbb{I}(U\otimes V)\sqrt{\rho}(U\otimes V)^\dagger\,\Pi^a \otimes \mathbb{I}\,(U\otimes V)\sqrt{\rho}].
    \end{align*}
    Since von-Neumann projectors are invariant under local unitary operations, the above quantity can be written as
    \begin{align*}
           \text{tr} [\Pi^a \otimes \mathbb{I}\sqrt{\rho}\,\Pi^a \otimes \mathbb{I}\,\sqrt{\rho}] 
          &= \text{tr} [\sqrt{\rho}\,\Pi^a \otimes \mathbb{I}\sqrt{\rho}\,\Pi^a \otimes \mathbb{I}]\\
          &=  \text{tr} [\sqrt{\rho}\,\Pi^a(\sqrt{\rho})].
    \end{align*}
\section{H-MIN for pure state}
\begin{theorem}
 For any bipartite pure state with Schmidt decomposition $\lvert \psi \rangle = \sum_i \sqrt{s_i}\lvert\alpha_i\rangle\otimes\lvert\beta_i\rangle$ with $s_i$ being the Schmidt co-efficients, $\lvert\alpha_i\rangle$ and $\lvert\beta_i\rangle$ are the orthonormal bases
\begin{equation}
  N_H(\lvert \psi \rangle \langle \psi \rvert) = \left(1-\sum_i s^{2}_i\right).
\end{equation}
\end{theorem}
\begin{proof}
Taking the density operator as $\rho = \sum_{ij}\sqrt{s_is_j}\lvert\alpha_i \rangle \langle\alpha_j \rvert \otimes \lvert\beta_i \rangle \langle\beta_j \rvert$,
the von Neumann projective measurements acting on the reduced state $\rho^a$ leaves the state invariant, that is $\Pi^a(\rho^a) = \rho^a$. Taking the set $ \{ \Pi^a_{k} \} = \{U\lvert\alpha_k \rangle \langle\alpha_k \rvert U^{\dagger} \}$ for any arbitary unitary operator $U$, we have
\begin{equation}
\rho^a = \sum_k  U\lvert\alpha_k \rangle \langle\alpha_k \rvert U^{\dagger}\rho^a U\lvert\alpha_k \rangle \langle\alpha_k \rvert U^{\dagger} = \sum_k s_k U \lvert \alpha_k \rangle \langle \alpha_k \rvert U ^\dagger
\end{equation}
since $s_k = \langle \alpha_k\lvert U^{\dagger}\rho^a U\lvert \alpha_k\rangle $, the eigenvalues of $\rho^a$. If the post-measurement state is given by
\begin{align*}
  \Pi^a(\rho) &= \sum_{ijk}\sqrt{s_i s_j}\,U \lvert \alpha_k \rangle \langle \alpha_k\lvert U^{\dagger}\rvert  \alpha_i\rangle \langle \alpha_j\lvert U\rvert  \alpha_k \rangle \langle \alpha_k\rvert  U^{\dagger}\otimes\lvert \beta_i\rangle \langle\beta_j \rvert 
  \end{align*}
  then
\begin{eqnarray*}
  \rho\,\Pi^a(\rho) &=& \sum_{i'j'ijk}\sqrt{s_i s_j s_{i'} s_{j'}}\,\lvert \alpha_{i'}\rangle \langle \alpha_{j'}\rvert U \lvert \alpha_k \rangle \langle \alpha_k\lvert U^{\dagger}\rvert  \alpha_i\rangle \langle \alpha_j\lvert U\rvert  \alpha_k\rangle  \langle \alpha_k\rvert  U^{\dagger} \\ && \otimes\lvert \beta_{i'}\rangle \langle \beta_{j'}\lvert \beta_i\rangle \langle\beta_j \rvert.
\end{eqnarray*}
Taking trace of the above expression 
\begin{align*}
\hbox{tr}[\rho\,\Pi^a(\rho)] &= 
           \sum_{ijk} s_is_j \langle \alpha_k \rvert U^\dagger \lvert \alpha_j \rangle 
                             \langle \alpha_i \rvert U \lvert \alpha_k \rangle 
                             \langle \alpha_k \rvert U^\dagger \lvert \alpha_i \rangle 
                             \langle \alpha_j \rvert U \lvert \alpha_k \rangle \\
                  &= \sum_k\left( \sum_i {s_i}\langle \alpha_k \rvert U^\dagger \lvert \alpha_i\rangle
                                    \langle \alpha_i \rvert U \lvert \alpha_k \rangle\right)^2 \\
                  &= \sum_k \bigg(\langle \alpha_k \rvert U^\dagger \rho^a  U \lvert \alpha_k \rangle \bigg)^2 = \sum_k s_k^2.
\end{align*}
Since $\rho$ is a pure state, $\rho^2 = \rho$ or $\rho= \sqrt{\rho} $. This completes the proof. 
 \end{proof}
  \section{H-MIN for mixed state}
A general bipartite state $\sqrt{\rho}$ on composite Hilbert space $\mathcal{H}^a \otimes \mathcal{H}^b$ can be written as 
\begin{equation}
  \sqrt{\rho} = \sum_{ij}\gamma_{ij} X_i \otimes Y_j\label{bipart}
\end{equation}
where $\Gamma = (\gamma_{ij})$ is a correlation matrix with real elements $\gamma_{ij} = \text{tr}(\sqrt{\rho}X_i\otimes Y_j)$. For any orthonormal basis $\{\lvert k\rangle : k = 1,2, \cdots, m\}, \lvert k\rangle \langle k\lvert = \sum_i a_{ki} X_i$ with $ a_{ki} = \text{tr}(\lvert k\rangle \langle k\rvert X_i)$\label{ak}. Defining a matrix $A = (a_{ij})$,~\text{we have}~$AA^t = \mathbb{I}_{m}$, where $A^t$ is the transpose of $A$.\\
\begin{theorem}
  For any bipartite state represented by eq.(\ref{bipart}), the H-MIN is bounded as
\begin{equation}
  N_H(\rho)  \leq \left(1- \sum_{i = 1}^{m-1} \mu_i\right)
\end{equation}
where $\{\mu_i, i = 1,2,\cdots, m^2\}$ are the eigenvalues of $\Gamma\Gamma^t$ listed in increasing order.
\end{theorem}
\begin{proof}
Let us begin with 
\begin{align*}
\text{tr}\big(\sqrt{\rho}\,\Pi^a(\sqrt{\rho})\big) &=\text{tr}\sum_{k}\sqrt{\rho}~(\Pi^a_k\otimes \mathbb{I}) \sqrt{\rho}~(\Pi^a_k\otimes \mathbb{I}) \\
 &= \text{tr}\sum_{i'j'ijk}(\gamma_{i'j'}~X_{i'} \otimes Y_{j'}~\lvert  k\rangle \langle k \rvert \otimes \mathbb{I}~\gamma_{ij}~X_i\otimes Y_j \lvert \,k\rangle \langle k \rvert \otimes \mathbb{I}) \\
 &= \text{tr}\sum_{i'j'ijk} \gamma_{i'j'}\gamma_{ij} X_i \lvert k\rangle \langle k\rvert X_{i'}  \lvert k\rangle \langle k\rvert \otimes Y_jY_{j'}\\
 &= \sum_{i'j'ijk} \gamma_{i'j'}\gamma_{ij}\langle k \rvert X_i \lvert k\rangle \langle k \rvert X_{i'}  \lvert k\rangle.\delta_{jj'} \\
 &= \sum_{i'ijk}a_{ki}\gamma_{ij}\gamma_{i'j}a_{ki'} \\
 &= \text{tr}(A\Gamma\Gamma^tA^t)
\end{align*}
To compute H-MIN we need to minimize the above quantity. The minimaization is required only when the reduced system $ \rho^a$ is degenerate. In the case of degenerate case, we adopt the following optimization procedure. Since $ \sum^{m^2}_{i =1}a_{ki}a_{k'i} = \delta_{kk'} $ and $a_{k1} = \text{tr}\big(\lvert k \rangle \langle k \rvert X_1 \big)= \frac{1}{\sqrt{m}}$, we have
\begin{equation}
  \sum^{m^{2}}_{i=2} a_{ki}a_{k'i}  =
  \begin{cases}
    \frac{m-1}{m} & \text{if} ~ k = k',\\
    -\frac{1}{m} & \text{if} ~ k \neq k'.
  \end{cases} \label{delta}
\end{equation}
 From this the real matrix $AA^t$ with eigenvalues 0 and 1 can be expressed as 
\begin{equation}
  AA^t = \frac{1}{m} 
  \begin{pmatrix}
    m-1  & -1   & ... &  -1\\
    -1   & m-1  & ... &  -1\\
    .    &  .   &  .  &   .\\
    .    &  .   &  .  &   .\\
    .    &  .   &  .  &   .\\
    -1   & -1   & ... &  m-1\label{a_matrix}
  \end{pmatrix}.
\end{equation}
This matrix can be diagonalised as $AA^t = UDU^t$, with $U$  being a real unitary matrix and
\begin{equation}
  D = 
  \begin{pmatrix}
    \mathbb{I}_{m-1} & 0\\
    0                & 0
  \end{pmatrix}.
\end{equation}
 Defining $B := U^tA = 
\begin{pmatrix}
  R\\
  0
\end{pmatrix}$
\begin{equation}
  BB^t = 
  \begin{pmatrix}
    \mathbb{I}_{m-1} & 0\\
    0                & 0
  \end{pmatrix}\label{b_matrx}.
\end{equation}
From the above equation, $RR^t = \mathbb{I}_{m-1}$ and the H-MIN can be written as
\begin{align*}
  N_H(\rho) &= (1-~^{\text{min}}_A \text{tr}\,A\Gamma\Gamma^tA^t) \\
            &= (1-~^{\text{min}}_R \text{tr}\,R\Gamma\Gamma^tR^t) \\
            &\leq \left(1- \sum^{m-1}_{i=1} \mu_i\right).
\end{align*}
Hence the theorem is proved.
\end{proof}
\begin{theorem}
  For any $2 \times n$ dimensional bipartite system $\rho$, H-MIN is given as
  \begin{equation*}
  N_H(\rho) =
  \begin{cases}
    (1-\mu_1) & \text{if} ~\textbf{x} = 0,\\
    (1-\hbox{tr}(\,A\Gamma\Gamma^tA^t)) & \text{if} ~\textbf{x} \neq 0    
  \end{cases}
  \end{equation*}
  \text{with}
 \begin{equation}
  A = \frac{1}{\sqrt{2}}
  \begin{pmatrix}
    1 & \frac{\textbf{x}}{\lVert \textbf{x}\lVert } \\
    1 & -\frac{\textbf{x}}{\lVert \textbf{x}\lVert }
 \end{pmatrix}\label{A_matrix}.
  \end{equation}
\end{theorem}
\begin{proof}
Let the reduced density matrix of subsystem be
\begin{equation}
\rho^a = \frac{1}{2}\mathbb{I}^a + \sum^{3}_{i=1}x_iX_i
\end{equation}
with $X_i \, \text{in terms of Pauli matrices as} \, X_i = \sigma_i/\sqrt{2},(i = 1,2,3)~\text{and}~\textbf{x} = (x_1,x_2,x_2)$.
For $\textbf{x} = 0$, the state $\rho^a = \mathbb{I}/2$. Then $\text{tr}(R\Gamma\Gamma^tR^t)$ = $\mu_1$. When $\textbf{x} \neq 0$, the von Neumann measurements which leave $\rho^a$ invariant are
\begin{equation}
  \Pi_1 = \frac{1}{2}\bigg(\mathbb{I}+ \frac{1}{\|\textbf{x}\|}\sum^{3}_{i=1}x_i\sigma_i\bigg),
\end{equation}\label{sp_pro1}
\begin{equation}
  \Pi_2 = \frac{1}{2}\bigg(\mathbb{I}- \frac{1}{\|\textbf{x}\|}\sum^{3}_{i=1}x_i\sigma_i\bigg).
\end{equation}\label{sp1}
From the above equations, $a_{ij} = \text{tr}(\Pi_iX_j)$ such that $a_{1i} = - a_{2i} = \frac{x_i}{\sqrt{2}\lVert \textbf{x}\rVert }$. With this we get
\begin{equation*}
  A = \frac{1}{\sqrt{2}}
  \begin{pmatrix}
    1 & \frac{\textbf{x}}{\lVert \textbf{x}\lVert } \\
    1 & -\frac{\textbf{x}}{\lVert \textbf{x}\lVert }
 \end{pmatrix}\label{A_matrix}.
\end{equation*}
This completes the proof.
\end{proof}
\section{Examples}
In this section we calculate H-MIN for some well known family of quantum states namely Bell diagonal state, Werner state and Isotropic state. 
\subsection{Bell diagonal state}
The Bloch vector representation of Bell diagaonal state is given as 
\begin{equation}
  \rho^{BD} = \frac{1}{4}\Big(\mathbb{I}\otimes \mathbb{I} + \sum^{3}_{i = 1} c_i(\sigma_i \otimes \sigma_i)\Big) \label{bell_diag}
\end{equation}
where $\textbf{c} = (c_1, c_2, c_3) $ are the correlation coefficients with $ -1 \leq c_i \leq 1$. Then
\begin{equation}
  \sqrt{\rho^{BD}} =  \frac{1}{4}\Big(\delta~\mathbb{I}\otimes \mathbb{I} + \sum^{3}_{i = 1} d_i(\sigma_i \otimes \sigma_i)\Big)
\end{equation}
where $\delta = \text{tr}(\sqrt{\rho^{BD}}) = \sum_i \sqrt{\lambda_i}$ and
\begin{equation*}
  d_1 = \sqrt{\lambda_1}- \sqrt{\lambda_2} + \sqrt{\lambda_3} - \sqrt{\lambda_4} 
\end{equation*}
\begin{equation*}
d_2 =  -\sqrt{\lambda_1}+ \sqrt{\lambda_2} + \sqrt{\lambda_3} - \sqrt{\lambda_4} 
\end{equation*}
\begin{equation*}
d_3 = \sqrt{\lambda_1} + \sqrt{\lambda_2} - \sqrt{\lambda_3} - \sqrt{\lambda_4}. 
\end{equation*}
Here $\lambda_i$ are the eigenvalues of the Bell diagonal state. Then H-MIN is computed as
\begin{equation}
  N_H(\rho^{BD}) = \left(1 - \frac{1}{4}(\delta^2 +\text{min}\{d^2_i\})\right).
\end{equation}
\begin{figure}[!ht]
\centering\includegraphics[width=0.7\linewidth]{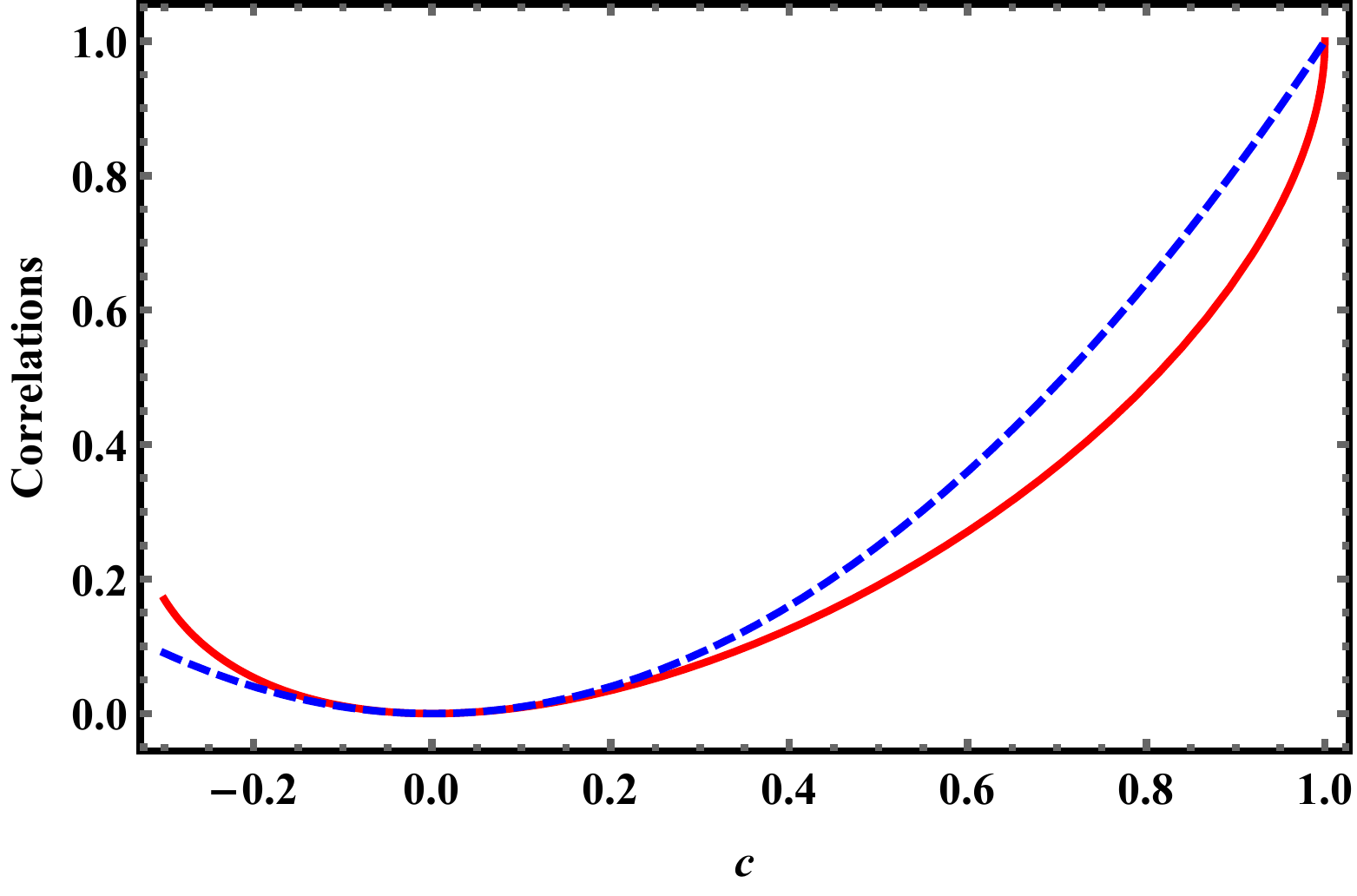}
\caption{(color online) H-MIN (solid) and MIN (dashed) for Bell diagonal state.}
\label{bell}
\end{figure}
\subsection{Isotropic state}
An $n\times n$ dimensional isotropic state is defined as \cite{isotropic} 
\begin{equation}
  \rho^{iso} = \frac{1-x}{n^2-1} \mathbb{I} +\frac{n^2~x - 1}{n^2 - 1}\lvert \phi \rangle \langle \phi \rvert
\end{equation}
where $ \lvert \phi \rangle =\frac{1}{\sqrt{n}} \sum_i\lvert ii\rangle $ with $ x \in [0, 1]$ for which H-MIN is calculated as 
\begin{equation}
  N_H(\rho^{iso}) = \frac{1}{n}\bigg(\sqrt{(n-1)x}-\sqrt{\frac{1-x}{n+1}}\bigg)^2.
\end{equation}
\begin{figure}[!ht]
\centering\includegraphics[width=0.7\linewidth]{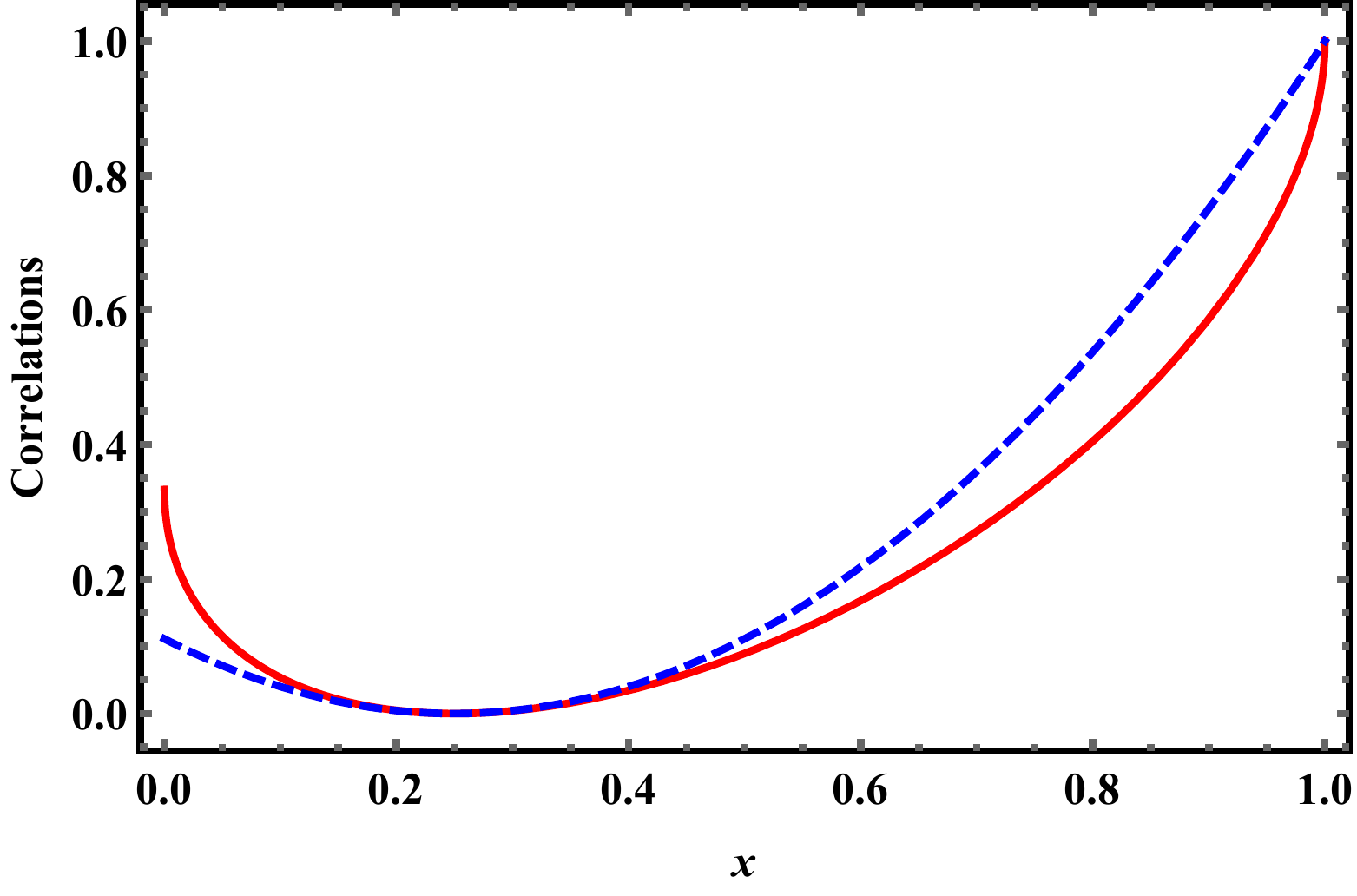}
\caption{(color online) H-MIN (solid) and MIN (dashed) for isotropic state with $n = 2$.}
\label{iso}
\end{figure}
\subsection{Werner state}
Werner state with $d\times d$ dimension can be represented as \cite{werner}
\begin{equation}
  \rho^{w} = \frac{d-x}{d^3-d}\mathbb{I} + \frac{xd-1}{d^3-d}\sum_{\alpha \beta}\lvert \alpha \rangle \langle\beta \rvert  \otimes\lvert \beta \rangle \langle \alpha  \rvert
\end{equation}
where $\sum_{\alpha \beta}\lvert \alpha \rangle \langle\beta \rvert  \otimes\lvert \beta \rangle \langle \alpha  \rvert$ is flip operator with $x \in [-1,1]$. In this case we have 
\begin{equation}
  N_H(\rho^w) = \frac{1}{2}\bigg(\frac{d-x}{d+1}-\sqrt{\frac{d-1}{d+1}(1-x^2)}\bigg).
\end{equation}
\begin{figure}[!ht]
\centering\includegraphics[width=0.7\linewidth]{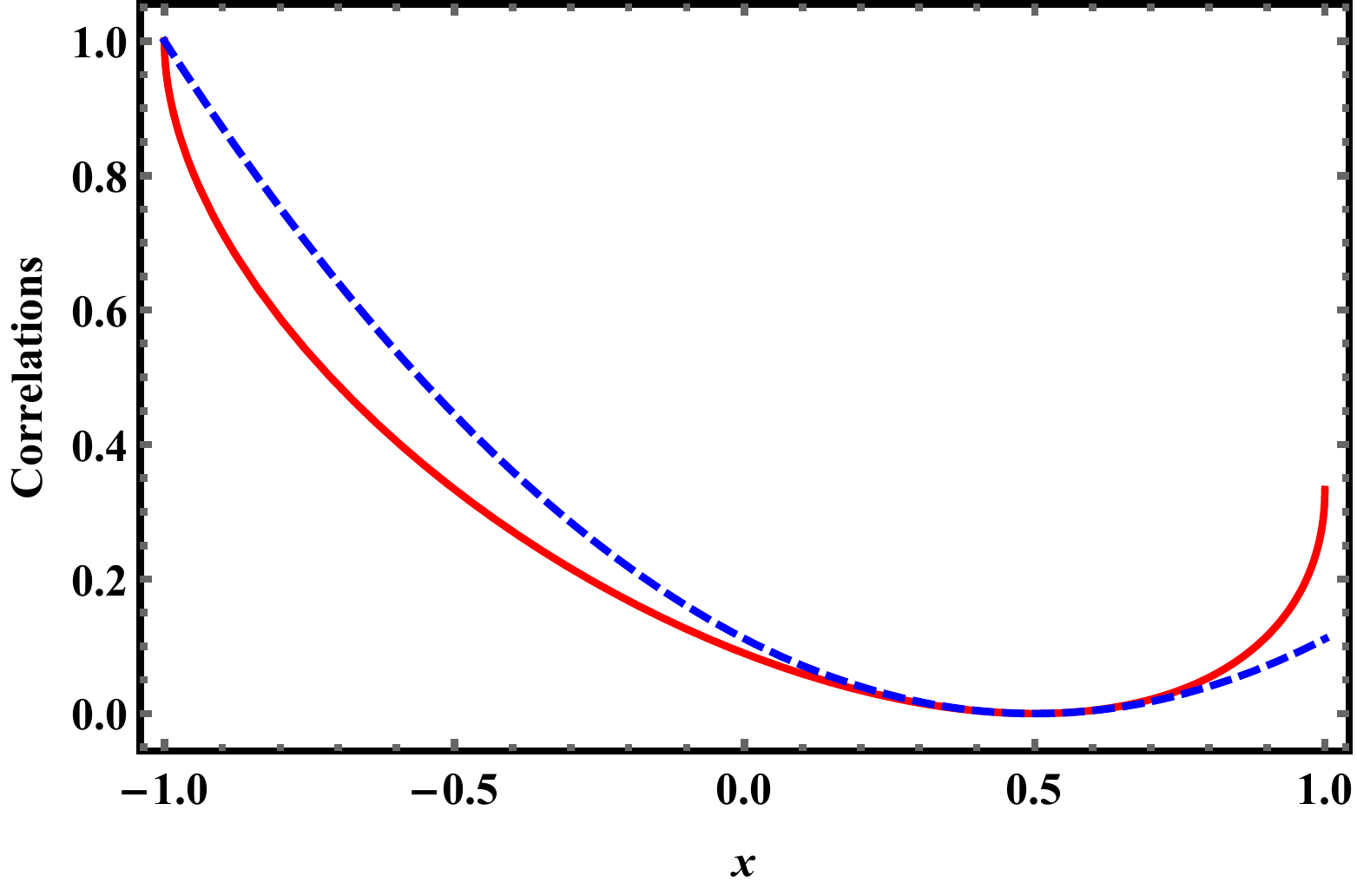}
\caption{(color online) H-MIN (solid) and MIN (dashed) for Werner state with $d = 2$.}
\label{wern}
\end{figure}
While plotting the above results, we multiply MIN and H-MIN by a factor 2, so that they vary from 0 to 1. In Fig.\ref{bell}, we plot the H-MIN for Bell diagonal state with $ c_1, c_2, c_3 = -c$ and the result is compared with MIN. It is clear that both H-MIN and MIN are qualitatively similar. We observe that both the quantities coincide for maximally entangled state ($ c = 1$) and diagonal state $ (c = 0)$. From Fig. \ref{iso} and Fig. \ref{wern} also we observe that H-MIN and MIN are quite consistent for isotropic and Werner states respectively.
\section{H-MIN based on weak measurements} 
In this section we define H-MIN based on weak measurements. A quantum measurement with any number of outcomes can be contructed as a sequence of two outcomes. Such weak measurement operators \cite{weak2} are defined as 
\begin{equation}
  \Omega_x = \tau_1 \Pi^1 + \tau_2 \Pi^2, ~~\Omega_{-x} = \tau_2 \Pi^1 + \tau_1 \Pi^2
\end{equation}\label{weak_op}
where $\tau_1 = \sqrt{\frac{1-\tanh{x}}{2}},~ \tau_2 = \sqrt{\frac{1+\tanh{x}}{2}}$ with $x \in \mathbb{R}$ being the strength of weak measurement and $\Pi^1,~\Pi^2$ are the two orthogonal projectors such that $\Pi^1 + \Pi^2 = \mathbb{I}$. Here $\Pi^1$ and $ \Pi^2$ can be decompossed as $ \Pi^1 = \sum^k_{i = 1} \Pi_i,\, \Pi^2 = \sum^n_{i = k+1} \Pi_i$, where $\{\Pi_i\}$ are von Neumann measurements \cite{weak_geo}. Weak measurements obey the relation $\sum_{k = \pm x} \Omega_k \Omega_k^{\dagger} = \mathbb{I}$. When $x\rightarrow \infty$ weak operators reduces to orthogonal projective measurements.
We define H-MIN based on weak measurement (WH-MIN) as 
\begin{equation}
  N_{w}(\rho) :=~^{\text{max}}_{\Omega} \lVert \sqrt{\rho}-\Omega(\sqrt{\rho}) \rVert^2
\end{equation}
where $\Omega(\sqrt{\rho}) = \sum_{k = \pm x} (\Omega_k \otimes \mathbb{I})~\sqrt{\rho}~(\Omega_k \otimes \mathbb{I})$.
\begin{lemma}
  For any bipartite state $\rho$, WH-MIN  
  \begin{equation}
    N_w(\rho) = (1-\tau)^2 N_H(\rho)
  \end{equation}
  where $ \tau = 2\tau_1\tau_2 = \sech\,x$.
\end{lemma}
\begin{proof}
\begin{align}
  \Omega(\rho) &= \sum_{k = \pm x}(\Omega_k \otimes \mathbb{I} )\rho (\Omega_k \otimes \mathbb{I})\nonumber \\ 
               &= 2\tau_1\tau_2 \rho + \sum^2_{i=1}(1- 2\tau_1\tau_2)(\Pi^i \otimes \mathbb{I}\,\rho\,\Pi^i \otimes  \mathbb{I} ) \nonumber\\ 
               &= 2\tau_1\tau_2 \rho + (1- 2\tau_1\tau_2 )\,\Pi^a(\rho) \nonumber \\ 
               &= \tau\,\rho + (1-\tau)\,\Pi^a(\rho)\label{weak_op}.
\end{align}
 With this, WH-MIN is computed as
\begin{align*}
  N_w(\rho)  &= \,^{\text{max}}_{\Omega} \text{tr}\big[\rho +  \Omega^2(\sqrt{\rho}) -2 \sqrt{\rho} \Omega(\sqrt{\rho})\big]\\
      &=  \,^{\text{max}}_{\Pi^a}\text{tr}\big[\rho + t^2 \rho + (1-t^2)\sqrt{\rho}\Pi(\sqrt{\rho}) - 2\big(t\rho + (1-\tau)\sqrt{\rho}\Pi^a(\sqrt{\rho})\big)\big]\\
      &= (\tau^2 - 2\tau +1) ^{\text{max}}_{\Pi^a}\,\text{tr}\big[\rho - \sqrt{\rho}\,\Pi^a(\sqrt{\rho})\big]\\
      &= (1-\tau)^2 \left(1-\, ^{\text{min}}_{\Pi^a}\,\text{tr}\big[\sqrt{\rho}\,\Pi^a(\sqrt{\rho})\big]\right) \\
      &= (1-\tau)^2 N_H(\rho).
  \end{align*}
  which completes the proof.
\end{proof}
In the asymptotic limit $x\rightarrow \infty$, $\Omega(\sqrt{\rho}) = \Pi^a(\sqrt{\rho})$ implying that $N_w(\rho) = N_H(\rho)$. In this sense $N_w(\rho)$ can be considered as the generalized form of $N_H(\rho)$. Defining $\rho_1 = \Omega(\sqrt{\rho})$, $\rho_2 = \Omega(\Omega(\sqrt{\rho})) = \Omega_2 (\sqrt{\rho})$ and so on, the sequential weak measurements on the state $\rho$ is denoted as
\begin{align}\nonumber
\rho_n &= \Omega_n(\sqrt{\rho}) \\ 
               &= \Omega_{n-1}\big(\Omega(\sqrt{\rho})\big) \nonumber\\
               &= \tau^n\sqrt{\rho} + \big(1- \tau^n\big)\Pi^a(\sqrt{\rho}) \label{rho_n}.
\end{align}
\begin{lemma}
  If $H^m_n(\rho) = \lVert \rho_{m}-\rho_n\rVert^2,$ then $ H^m_n = (\tau^m - \tau^n)^2 \,\mathcal{I}(\sqrt{\rho},\Pi^a\otimes \mathbb{I})$, where $ \mathcal{I}(\sigma,K)$ is the Wigner-Yanase skew information.
\end{lemma}
\begin{proof}
 We have
\begin{align*}
    H^m_n(\rho) &= \| \rho_{m} - \rho_n \|^2 \\
                &= \text{tr}[\rho^2_m + \rho^2_n - 2 A_{mn}]
\end{align*}  
where $A_{mn} = \rho_{m}\,\rho_n$ which is computed as
\begin{align*}
    A_{mn} &=[\tau^{m}\sqrt{\rho} + (1-\tau^{m})\Pi^a(\sqrt{\rho}) ][\tau^n\sqrt{\rho}+(1-\tau^n)\Pi^a(\sqrt{\rho})]\\
    &= \tau^{m+n}\rho +(1-\tau^{m+n})\sqrt{\rho}\,\Pi^a(\sqrt{\rho})
  \end{align*}
we make use of the identity $\text{tr}[\left(\Pi^a(\sqrt{\rho})\right)^2] = \text{tr}[\sqrt{\rho}\,\Pi^a(\sqrt{\rho})] $.
\begin{align*}              
       H^m_n(\rho)          &= \text{tr}[\rho(\tau^{2m} + \tau^{2n} - 2\tau^{m+n}) - \sqrt{\rho}\,\Pi^a(\sqrt{\rho})(\tau^{2m} + \tau^{2n} - 2\tau^{m+n})]\\
                &= (\tau^m - \tau^n)^2 \text{tr}[\rho  - \sqrt{\rho}\,\Pi^a(\sqrt{\rho})]\\
                &= (\tau^m - \tau^n)^2 \,\mathcal{I}(\sqrt{\rho},\Pi^a\otimes \mathbb{I}).
  \end{align*}  
 Hence the lemma is proved. 
 \end{proof} 
  As a special case with $ m =0$,
  \begin{equation}
  H^0_n(\rho) = \lVert \sqrt{\rho} -\rho_n \rVert^2 
\end{equation}\label{n_weak}
is the Hellinger distance between the states $ \rho$ and $ \rho_n$. Since $\cosh x = \sum^{\infty}_{k=0}\frac{x^{2k}}{(2k)!}$, $ \lim_{n \rightarrow \infty} \,(\cosh x)^n = \infty$, implying that $\tau^n = 0$ in the same limit. In other words, 
\begin{equation}
  \lim_{n\rightarrow \infty}\,\rho_n = \Pi^a(\sqrt{\rho})
\end{equation}
implying that infinte application of sequential weak measurement on the state $ \rho$ is equivalent to the post-measured state, and hence we have the relation
\begin{eqnarray}
 N_H(\rho) = \, ^{\text{max}}_{\Pi^a}\,\lim_{n\rightarrow \infty} H^0_n(\rho). 
\end{eqnarray}
In what follows we look at $ H^0_n(\rho)$ for a maximally entangled pure state $ \lvert \psi \rangle = \frac{1}{\sqrt{2}}\left(\lvert 00 \rangle + \lvert 11 \rangle\right)$. In Fig. \ref{weakpic}(i) we plot $ H^0_n$ as a function as $x$, the strength of weak measurement, for fixed $n$. Fig. \ref{weakpic}(ii) shows $ H^0_n$ as function of $n$ for fixed strength. It is clear that $ H^0_n$ increases with $x$($n$) to reach its maximum value of 0.5 as $ x (n) \rightarrow \infty$, as we intuitively expect.
\begin{figure}[!ht]
\includegraphics[width=0.51\linewidth]{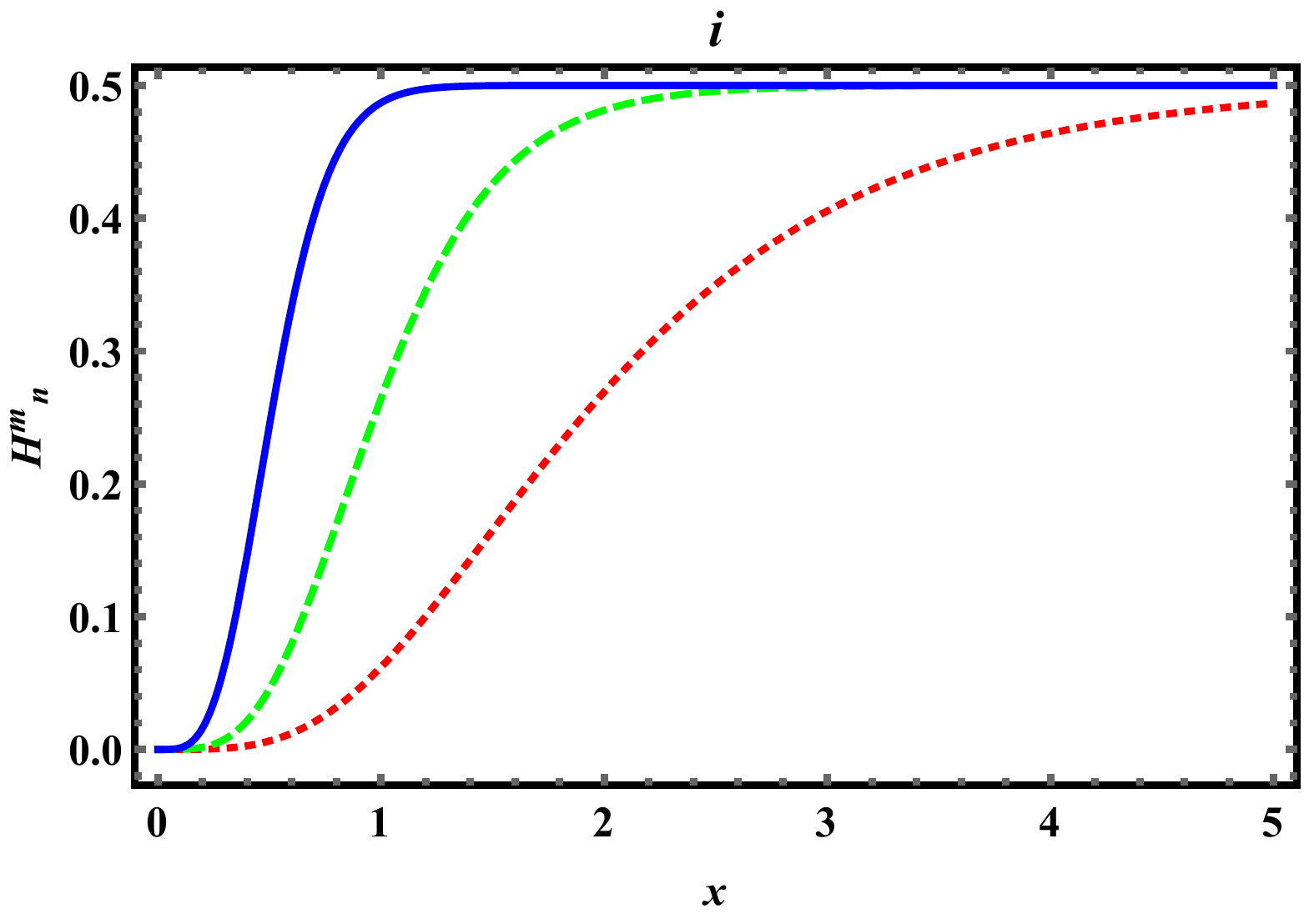}
\includegraphics[width=0.51\linewidth]{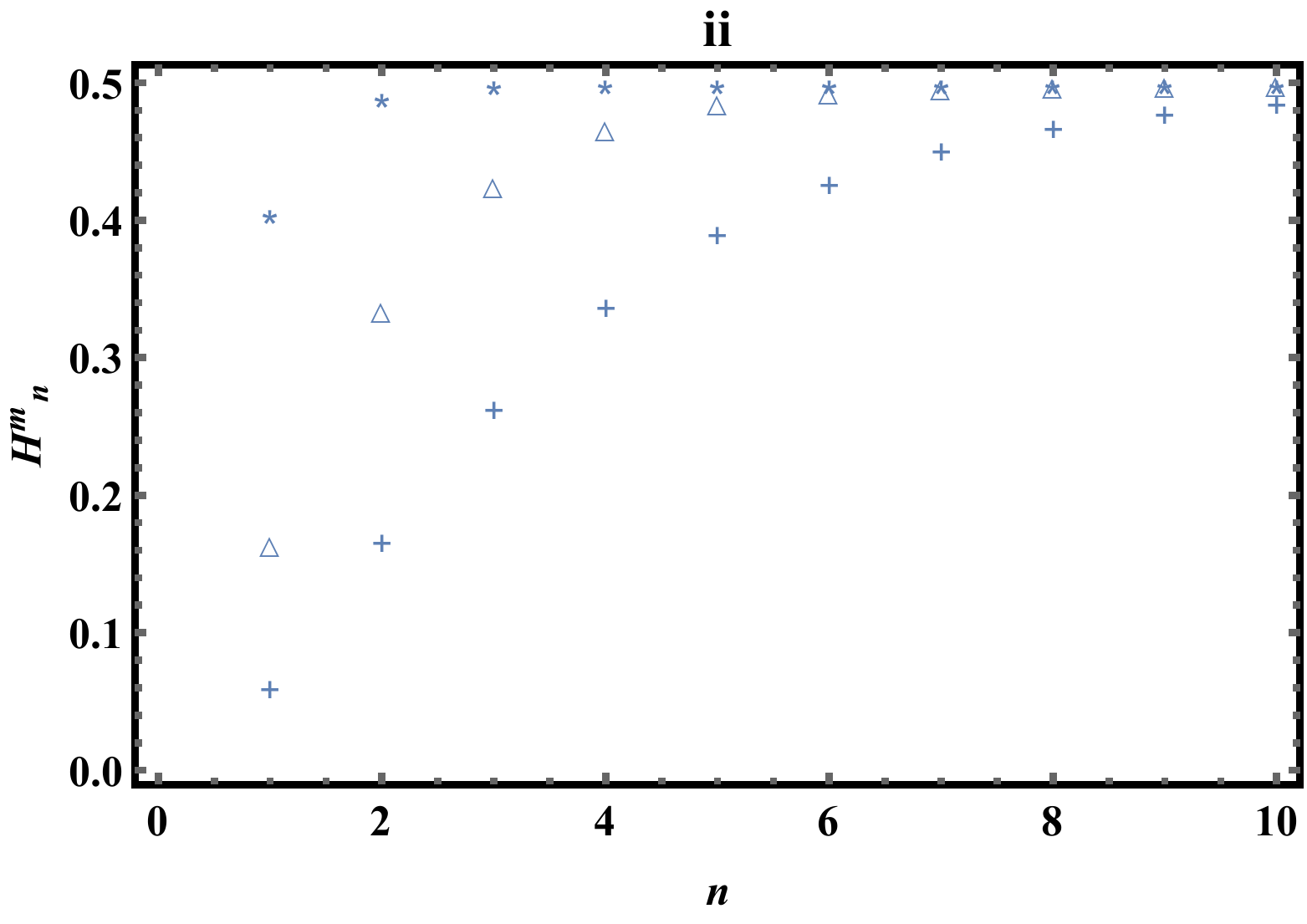}
\caption{(color online) (i) $H^0_n$ for $ n = 1\,(\text{dotted}), 3\,(\text{dashed}), 10 \,(\text{solid})$, (ii) $H^0_n$ for $ x = 1 (+), 1.5(\triangle), 3 (*).$}
\label{weakpic}
\end{figure}
\section{Relation with other correlation measures}
In this section we are establishing relations between H-MIN and its weak counterpart with other correlation measures.
\subsection{Skew-MIN} 
Skew information of state $\rho$ on a hermitian operator $ \mathcal{O}$ is defined as \cite{skew_yanse}
\begin{align*}
  \mathcal{I}(\rho,\mathcal{O}) = -\frac{1}{2}\text{tr}[\sqrt{\rho},\mathcal{O}]^2
\end{align*}
where $[a,b]$ stands for the commutation between the operators $a$ and $b$.
MIN based on skew information (Skew-MIN) is defined as \cite{skew}
\begin{align}\nonumber
  \mathcal{N}_s(\rho) &= ~^{\text{max}}_{\Pi^a} \sum_k \mathcal{I}(\rho,\Pi^a_k \otimes \mathbb{I}) \label{skew_min}\\ \nonumber
 &=\, -\frac{1}{2}~^{\text{max}}_{\Pi^a} \sum_k \text{tr} [\sqrt{\rho},\Pi_k^a\otimes \mathbb{I}]^2 \\ \nonumber
 &=\, 1 - \,^{\text{min}}_{\Pi^a} \text{tr}\left(\sqrt{\rho}\,\Pi^a(\sqrt{\rho})\right)\\
 &=\, N_H(\rho).
\end{align}
This implies that skew-MIN and H-MIN are one and the same. Defining weak skew-MIN (WS-MIN) as 
\begin{align}
\mathcal{N}_w(\rho) =&  -\,\frac{1}{2} ~^{\text{max}}_{\Omega}\,\text{tr} [\sqrt{\rho},\Omega\otimes\mathbb{I}]^2\\ \nonumber
                    =&  \,\left(1 -\, ^{\text{min}}_{\Omega}\text{tr} [\sqrt{\rho}\,(\Omega\otimes\mathbb{I})\sqrt{\rho}\,(\Omega\otimes\mathbb{I})]\right)\\ \nonumber
                    =&  \,^{\text{max}}_{\Omega}\lVert \sqrt{\rho} - \Omega(\sqrt{\rho})\rVert^2\\ \nonumber
                    =& \,N_w(\rho)
\end{align} \label{weak_skew} 
showing the equivalence of WS-MIN and WH-MIN as well.                                        
\subsection{Affinity-MIN}
Affinity is another quantity to measure the closeness  between two probability distributions, which is defined as \cite{heli_dist}
\begin{eqnarray*}
  \mathcal{A}(p,q) = \sum_x \sqrt{p(x)}\sqrt{q(x)}
\end{eqnarray*}
where $ p(x)$ and $q(x)$ are probability distributions. Upon extending this measure to the quantum regime,
by replacing the probability distributions with density matrices, the closeness between quantum states $ \rho$ and $ \sigma$ is quantified as 
\begin{equation*}
  \mathcal{A}(\rho,\sigma) = \text{tr}(\sqrt{\rho}\sqrt{\sigma}).
\end{equation*}
With this one can define MIN based on affinity (Affinity-MIN) as \cite{affinity}
\begin{align}
  N_{\mathcal{A}}(\rho) &= \left(1-\,^\text{min}_{\Pi^a} \text{tr}[\sqrt{\rho}\sqrt{\Pi^a(\rho)}]\right).
\end{align}
If $ \sqrt{\Pi(\rho)} = \Pi(\sqrt{\rho})$, which is the case for eigen projectors, we have 
\begin{equation*}
   N_{\mathcal{A}}(\rho) = N_H(\rho).
\end{equation*}
\section{Conclusion}
In this paper, we have proposed Hellinger distance based measurement-induced nonlocality (H-MIN) to quantify the quantum correlation of bipartite states. Since H-MIN involves square root of $\rho$, it naturally resolves the local ancilla problem. A closed formula of H-MIN for $2 \times n$ system is obtained, from which we have computed the same for few well-known families of quantum states. The computed results are shown to be consistent with the corresponding MIN. Further, our analysis reveals the equivalence of H-MIN and Skew-MIN, and is unaffected by the weak measurement as well. It is shown that repeated weak measurement on a given state will asymptotically lead to post-measured state.
\section*{Acknowledgement}
The author RM thank the Council of Scientific and Industrial Research (CSIR), Government of India for the financial support under Grant No. 03(1444)/18/EMR-II.
\bibliographystyle{99}

\end{document}